\documentclass[ijoo,nonblindrev]{informs-ijoo}

\OneAndAHalfSpacedXI

\usepackage{natbib}
\usepackage{optidef}
 \bibpunct[, ]{(}{)}{,}{a}{}{,}%
 \def\bibfont{\small}%
\usepackage{algorithm,url,lipsum,changepage}
\usepackage[noend]{algpseudocode}
\usepackage{fullwidth}
\usepackage{enumitem}
\usepackage{pifont}
\usepackage{verbatim}
\usepackage{amsmath}
\usepackage{amssymb}
\usepackage{mathtools}

\DeclarePairedDelimiter\floor{\lfloor}{\rfloor}
\TheoremsNumberedThrough     
\ECRepeatTheorems

\EquationsNumberedThrough    

\MANUSCRIPTNO{}
\newtheorem{lem}{Lemma}

\newtheorem{thm}[lem]{Theorem}

\begin{document}


\RUNAUTHOR{Udwani}

\RUNTITLE{Maximizing Multiple Monotone Submodular Functions}

\TITLE{Multi-objective Maximization of Monotone Submodular Functions  with Cardinality Constraint}

\ARTICLEAUTHORS{%
\AUTHOR{Rajan Udwani}
\AFF{UC Berkeley, IEOR, \EMAIL{rudwani@berkeley.edu}} 
} 

\ABSTRACT{%
	We consider the problem of multi-objective maximization of monotone submodular functions subject to cardinality constraint, often formulated as $\max_{|A|=k}\min_{i\in\{1,\dots,m\}}f_i(A)$. While it is widely known that greedy methods work well for a single objective, the problem becomes much harder with multiple objectives. In fact, Krause et al.\ (2008) showed that when the number of objectives $m$ grows as the cardinality $k$ i.e., $m=\Omega(k)$, the problem is inapproximable (unless $P=NP$). On the other hand, when $m$ is constant Chekuri et al.\ (2010) showed a randomized $(1-1/e)-\varepsilon$ approximation with runtime (number of queries to function oracle) the scales as $n^{m/\varepsilon^3}$. 

We focus on finding a fast algorithm that has (asymptotic) approximation guarantees even when $m$ is super constant. We first modify the algorithm of Chekuri et al.\ (2010) to achieve a $(1-1/e)$ approximation for $m=o(\frac{k}{\log^3 k})$. This demonstrates a steep transition from constant factor approximability to inapproximability around $m=\Omega(k)$. Then using Multiplicative-Weight-Updates (MWU), we find a much faster $\tilde{O}(n/\delta^3)$ time asymptotic $(1-1/e)^2-\delta$ approximation. While the above results are all randomized, we also give a simple deterministic $(1-1/e)-\varepsilon$ approximation with runtime $kn^{m/\varepsilon^4}.$ 
Finally, we run synthetic experiments using Kronecker graphs and find that our MWU inspired heuristic outperforms existing heuristics.
}%


\KEYWORDS{multiple objectives, monotone submodular functions, approximation algorithms} 

\maketitle

%



\section{Introduction}
Many well known objectives in combinatorial optimization exhibit two common properties: the marginal value of any given element is non-negative and it decreases as more and more elements are selected. The notions of submodularity and monotonicity \footnote{A set function $f:2^N\rightarrow \mathbb{R}$ on the ground set $N$ 
	is called submodular when
	$f(A+a)-f(A)\leq f(B+a)-f(B) \text{  for all } B\subseteq A\subseteq N \text{ and } a\in N\setminus A.$. The function is monotone if $f(B)\leq f(A) \text{ for all } B\subseteq A$. W.l.o.g., assume $f(\emptyset)=0$. Combined with monotonicity this implies non-negativity. } capture these properties, resulting in the appearance of constrained monotone submodular maximization in a wide and diverse array of modern applications in machine learning and optimization. This includes feature selection (\cite{motive1,feature1}), network monitoring (\cite{sense3}), news article recommendation (\cite{news}), sensor placement and information gathering (\cite{sense0,sense1,sense2,sense4}), viral marketing and influence maximization (\cite{kempeinflu,influ}), document summarization (\cite{docsum}) and crowd teaching (\cite{crowd}).  

In this paper, we are interested in scenarios where multiple objectives, all monotone submodular, need to be simultaneously maximized subject to a cardinality constraint. This problem has an established line of work in both machine learning \citep{main}  and the theory community \citep{swap}. As an example application, in robust experimental design one often seeks to maximize a function $f_{\theta}$, which is monotone submodular for every value of $\theta$. The function is very sensitive to the choice of $\theta$ but the parameter is unknown a priori and estimated from data. Therefore, one possible approach to finding a robust solution is to maximize the function $min_{\theta\in \Theta} f_{\theta}(.)$, where $\Theta$ is a set that captures the uncertainty in $\theta$. If $\Theta$ is assumed to be a finite set of discrete values \citep{main} we have an instance of multi-objective monotone submodular maximization. More generally, we consider the following problem,
$$MO_1: \max_{A\subseteq N, |A|\leq k}\, 
\min_{i\in\{1,2,\dots,m\}} f_i(A),$$
where $f_i(.)$ is monotone submodular for every $i\in\{1,\dots,m\}$. The problem also has an alternative formulation due to \cite{swap,swapfocs}, which we discuss later. Broadly speaking, there are two ways in which this framework has been applied --

\noindent\textbf{When there are several natural criteria that need to be simultaneously optimized: }such as in network monitoring, sensor placement and information gathering \citep{sense0,sense3,sense4,main}. For example in the problem of intrusion detection \citep{sense0}, one usually wants to maximize the likelihood of detection while also minimizing the time until intrusion is detected and the population affected by intrusion. The first objective is often monotone submodular and the latter objectives are monotonically decreasing supermodular functions \citep{sense3,sense4}. Therefore, the problem is often formulated as an instance of cardinality constrained maximization with a small number of submodular objectives.  

\noindent \textbf {When looking for solutions robust to the uncertainty in objective: }such as in feature selection \citep{main,ml}, variable selection and experimental design \citep{main}, robust influence maximization \citep{influ}. In these cases, there is often inherently just a single submodular objective which is highly prone to uncertainty either due to dependence on a parameter that is estimated from data, or due to multiple possible scenarios that each give rise to a different objective. Therefore, one often seeks to optimize over the worst case realization of the uncertain objective, resulting in an instance of multi-objective submodular maximization. 

In some applications the number of objectives is given by the problem structure and can be larger even than the cardinality parameter. However, in applications such as robust influence maximization, variable selection and experimental design, the number of objectives is a design choice that trades off optimality with robustness.

\subsection{Related Work}
The problem of maximizing a monotone submodular function subject to a cardinality constraint, $$P_0:= \max_{A\subseteq N, |A|\leq k} f(A),$$ goes back to the work of \cite{nem,nem1}, where they showed that the greedy algorithm gives a guarantee of $(1-1/e)$ and this is best possible in the value-oracle model. Later, \cite{hard} showed that this is also the best possible approximation unless P=NP. While this settled the hardness and approximability of the problem, finding faster approximations remained an open line of inquiry. Notably, \cite{fast} found a faster algorithm for $P_0$ that improved the quadratic $O(nk)$ query complexity of the classical greedy algorithm to nearly linear complexity by trading off on the approximation guarantee. Subsequently, \cite{lazyfast} found an even faster randomized algorithm. 

For the more general problem $\max_{A\in\mathcal{I}} f(A)$, where $\mathcal{I}$ is the collection of independent sets of a matroid; \cite{vond1,welf} in a breakthrough, achieved a $(1-1/e)$ approximation by (approximately) maximizing the \textit{multilinear extension} of submodular functions, followed by suitable rounding. 
Based on this framework, tremendous progress was made over the last decade for a variety of different settings (\cite{vond1,welf,multiuni,hard1,vond,swap}). 

In the multi-objective setting, \cite{main} amalgamated various applications and formally introduced the following problem,
$$P_1=\max_{A\subseteq N, |A|\leq k}\, 
\min_{i\in\{1,2,\dots,m\}} f_i(A),$$
where $f_i(.)$ is monotone submodular for every $i$. They call this the Robust Submodular Observation Selection (RSOS) problem and show that in general the problem is inapproximable (no non-trivial approximation possible) unless $P=NP$. Consequently, they proceeded to give a bi-criterion approximation algorithm, called SATURATE, which achieves the optimal answer by violating the cardinality constraint. Note that their inapproximability result only holds when  $m=\Omega(k)$. Another bi-criterion approximation was given more recently in \cite{chen}. 

On the other hand, \cite{swap} showed 
a randomized $(1-1/e)-\varepsilon$ approximation for constant $m$ in the more general case of matroid constraint, as an application of a new technique for rounding over a matroid polytope, called \emph{swap rounding}. The runtime scales as $O(n^{m/\varepsilon^3} + mn^8)$. \footnote{\label{runtime}It may be possible to improve the $n^8$ term in the runtime to $n^2$ by leveraging the ideas in \cite{fast}. However, this does not immediately follow from existing results which are known only for maximizing a single function.} Note, \cite{swap} consider a different but equivalent formulation of the problem that stems from the influential paper on multi-objective optimization by \cite{papad}. The alternative formulation, which we review in Section \ref{prelim}, is the reason we call this a multi-objective maximization problem (same as \cite{swap}). For the special case of cardinality constraint (which will be our focus here), \cite{ipcofull} recently showed that the greedy algorithm can be generalized to achieve a \emph{deterministic} $1-1/e-\varepsilon$ approximation for the special case of bi-objective maximization. Their runtime scales as $n^{1+ 1/\varepsilon}$ and $\varepsilon\leq 1/2$. To the best of our knowledge, when $m=o(k)$ no  constant factor approximation algorithms or inapproximability results were known prior to this work. 
\subsection{Our Contributions}
Our focus here is on the regime $m=o(k)$. This setting is essential to understanding the approximability of the problem for super-constant $m$ and includes several of the applications we referred to earlier. For instance, in network monitoring and sensor placement, the number of objectives is usually a small constant \citep{main,sense3}. For robust influence maximization, the number of objectives depends on the underlying uncertainty but is often small \citep{influ}, and in settings like variable selection and experimental design \citep{main}, the number of objectives considered is a design choice. We show three algorithmic results with asymptotic approximation guarantees for $m=o(k)$. 

\textbf{1. Asymptotically optimal approximation:} We give a $(1-1/e-\varepsilon)(1-\frac{m}{k\varepsilon^3})$ approximation, which for $m=o\big(\frac{k}{\log^3 {k}}\big)$ and $\varepsilon=\min\{\frac{1}{10\ln m},\sqrt[4]{\frac{m}{k}}\}$ tends to $1-1/e$ as $k\to \infty$. The algorithm is randomized and outputs such an approximation w.h.p. Observe that this implies a steep transition around $m$, due to the inapproximability result (to within any non-trivial factor) for $m=\Omega(k)$. 

We obtain this via extending the matroid based algorithm of \cite{swap,swapfocs}, which relies on the \emph{continuous greedy} approach, resulting in a runtime of $\tilde{O}(mn^8)$. Note that there is no exponential term in the runtime, unlike the result from \cite{swap}. 
The key idea behind the result is quite simple, and relies on exploiting the fact that we are dealing with a cardinality constraint, far more structured than matroids. 

\textbf{2. Faster nearly linear time algorithm:}
In practice, $n$ can range from tens of thousands to millions (\cite{sense0,sense3}), which makes the above runtime intractable. To this end, we develop a fast $O(\frac{n}{\delta^3}\log m\log \frac{n}{\delta})$ time $(1-1/e)^2(1-m/k\varepsilon^3)-\varepsilon-\delta$ approximation. Under the same asymptotic conditions as above, the guarantee simplifies to $(1-1/e)^2-\delta$.
We achieve this via the Multiplicative-Weight-Updates (MWU) framework, which replaces the bottleneck continuous greedy process. This costs us the additional factor of $(1-1/e)$ in the guarantee but allows us to leverage the runtime improvements for $P_0$ achieved in \cite{fast,lazyfast}. 

MWU has proven to be a vital tool in the past few decades (\cite{mwu1,mwu3,mwu4,mwu8,mwu9,plotkin,arora}). Linear functions and constraints have been the primary setting of interest in these works, but recent applications have shown its usefulness when considering non-linear and in particular submodular objectives (\cite{azar,chekurimwu}). Unlike these recent applications, we instead apply the MWU framework in vein of the Plotkin-Shmoys-Tardos (PST) scheme for linear programming (\cite{plotkin}), essentially showing that the non-linearity only costs us a another factor of $(1-1/e)$ in the guarantee and yields a nearly linear time algorithm. Subsequent to a preliminary version of this work we discovered independent work by \cite{chen} where they applied the MWU framework for submodular objectives in a manner resembling the PST framework. However, they use it to derive a new bi-criterion approximation (with bounds similar to the result in \cite{main}), as opposed to a constant factor approximation.

\textbf{3. Deterministic approximation for small $m$:}
While the above results are all randomized, we also show a simple greedy based deterministic $1-1/e-\varepsilon$ approximation with runtime $kn^{m/\varepsilon^4}$. This follows by establishing an upper bound on the increase in optimal solution value as a function of cardinality $k$, which also resolves a weaker version of a conjecture posed in \cite{ipcofull}.
\\

\noindent \textbf{Corollaries and immediate implications of our results.} \\
\emph{1. Curvature optimal algorithm:} The curvature of a submodular function  $f$ is a parameter $c\in[0,1]$ such that for any set $S\subset N$ and element $i\in N\backslash S$, 
\[f(S\cup \{j\})-f(S)\geq (1-c)f(\{j\}). \]
Let $c_i$ denote the curvature for function $f_i$ and $c$ the maximum curvature $\max_{i\in[m]}c_i$. Using the curvature optimality established by \cite{curve}, our results naturally extend to yield a curvature optimal $\frac{1}{c}(1-e^{-c})$ asymptotic approximation for $m=o\big(\frac{k}{\log^3 {k}}\big)$. This, for instance, implies an asymptotically optimal algorithm for the special case where the objectives are all linear functions. 
\\~\\
\noindent \emph{2. Improved approximation for robust submodular maximization:} Given a monotone submodular function $f$, \cite{main} introduced the following robust maximization problem:
\[RO:=\max_{A\subseteq N, |A|\leq k}\, 
\min_{B\subseteq A, |B|\leq \tau} f(A\backslash B), \]
where one seeks a set $A$ of size at most $k$, that has the maximum function value after removal of any subset of up to $\tau$ elements. \cite{main} observed that this problem can be reduced to an instance of $P_1$ with $O(n^{\tau})$ monotone submodular objectives given by $f(\cdot\backslash B)$ for every $B\subset N$ of size $\tau$. Using this idea they proposed a bi-criterion approximation for the problem with runtime scaling as $n^{\tau}$. Later, \cite{ipcofull} showed that to get a nearly optimal $(1-1/e)-\varepsilon$ approximation for $RO$, it is enough to consider an instance of $P_1$ with $2^{O(\tau\log \tau)}$ functions. Consequently, the algorithm of \cite{swap} for $P_1$, previously allowed a $(1-1/e)-\varepsilon$ approximation for $RO$ with constant $\tau$. 

Our new result for $P_1$ directly extends the $(1-1/e)-\varepsilon$ approximation for $RO$ in \cite{ipcofull} to super-constant $\tau$, as long as the number of objectives $2^{O(\tau\log \tau)}$ is $o(k/\log^3 k)$. This holds in particular for $\tau=\frac{o(\log k)}{\log\log k}$. Previously, the best known approximation for super constant $\tau$ was 0.387 (\cite{ipcofull}, \cite{bogu}). 

\textbf{Outline:} We start with definitions and preliminaries in Section 2, where we also review relevant parts of the algorithm in \cite{swap} that are essential for understanding the results here. In Section 3, we state and prove the main results. 
Since the guarantees we present are asymptotic and technically converge to the constant factors indicated as $k$ becomes large, in Section 4 we test the performance of a heuristic, closely inspired by our MWU based algorithm, on Kronecker graphs \cite{kron} of various sizes and find improved performance over previous heuristics even for small $k$ and large $m$. 

\section{Preliminaries}\label{prelim}
\subsection{Definitions \& review}

We work with a ground set $N$ of $n$ elements and recall that we use $P_0$ to denote the single objective (classical) problem. \cite{nem,nem1} showed that the natural \emph{greedy algorithm} for $P_0$ achieves a guarantee of $1-1/e$ for $P_0$ and that this is best possible. The algorithm can be summarized as follows --

\emph{Starting with $\emptyset$, at each step add to the current set an element which adds the maximum marginal value until $k$ elements are chosen.} 

Formally, given set $A$ the marginal increase in value of function $f$ due to inclusion of set $X$ is, $$f(X|A)=f(A\cup X)-f(A).$$
Let $\beta(\eta)=1-\frac{1}{e^{\eta}}\in [0,1-1/e]$ for $\eta \in[0,1]$. Note that $\beta(1)=(1-1/e)$. Further, for $k'\leq k$, 
\begin{equation}
\beta(k'/k)= (1-e^{1-k'/k}/e)\geq (1-1/e)k'/k. \label{simfact}
\end{equation}
This function appears naturally in our analysis and will be useful for expressing approximation guarantees.

%
%
%
We use the notation $\textbf {x}_S$ for the support vector of a set $S$ (1 along dimension $i$ if $i\in S$ and 0 otherwise). We also use the short hand $|\textbf{x}|$ to denote the $\ell_1$ norm of a vector $\textbf {x}$ i.e.,
\[ |\textbf{x}|:=|\textbf{x}|_1. \]
 Given $f:2^{N}\to \mathbb{R}$, recall that its \emph{multilinear extension} over ${\bf x}=\{x_1,\dots,x_n\}\in [0,1]^{n}$ is defined as, $$F({\bf x})=\sum_{S\subseteq N}f(S)\prod_{i\in S}x_i\prod_{j\not\in S}(1-x_j).$$
The function can also be interpreted as the expectation of function value over sets obtained by including element $i\in N$ independently with probability $x_i, \forall i$.  Further, given two vectors ${\bf x},{\bf y}\in[0,1]^n$, let ${ \bf x}\vee \bf{y}$ denote the component wise maximum. Then we define marginals for $F$ as, $$F({\bf x}|{\bf y})=F({ \bf x}\vee {\bf y})- F({\bf y}).$$ 
$F$ acts as a natural replacement for the original function $f$ in the \emph{continuous greedy algorithm} (\cite{vond1}). Like the greedy algorithm, the \emph{continuous} version always moves in a feasible direction that best increases the value of function $F$. While evaluating the exact value of this function and its gradient is naturally hard in general, for the purpose of using this function in optimization algorithms, approximations to marginal values  obtained using random sampling suffice. For completeness, we include a more formal description in Appendix \ref{estimatingF}.


Now, we briefly discuss another formulation of the multi-objective maximization problem, call it $P_2$, introduced in \cite{swap}. In $P_2$ we are given a target value $V_i$ (positive real) with each function $f_i$ and the goal is to find a set $S^*$ of size at most $k$, such that $f_i(S^*)\geq V_i,\, \forall i\in\{1,\dots,m\}$ or certify that no $S^*$ exists. More feasibly one aims to efficiently find a set $S$ of size $k$ such that $f_i(S)\geq \alpha V_i$ for all $i$ and some factor $\alpha$, or certify that there is no set $S^*$ of size $k$ such that $f_i(S^*)\geq V_i,\, \forall i$. Observe that w.l.o.g.\ we can assume $V_i=1, \forall i$ (since we can consider functions $f_i(.)/V_i$ instead) and therefore $P_2$ is equivalent to the decision version of $P_1$: \emph{Given $t>0$, find a set $S^*$ of size at most $k$ such that $\min_{i}f_i(S^*)\geq t$, or give a certificate of infeasibility.}  

%

When considering formulation $P_2$, \emph{since we can always consider the modified submodular objectives $\min\{f_i(.),V_i\}$, we w.l.o.g.\ assume that $f_i(S)\leq V_i$ for every set $S$ and every function $f_i$.} Finally,  for both $P_1,P_2$ we use $S_k$ to denote an optimal/feasible set (optimal for $P_1$, and feasible for $P_2$) to the problem and $OPT_k$ to denote the optimal solution value for formulation $P_1$. 
We now give an overview of the algorithm from \cite{swap} which is based on $P_2$. To simplify the description we focus on cardinality constraint, even though it is designed more generally for matroid constraint. We refer to it as \textbf{Algorithm 1} and it has three stages. Recall, the algorithm runs in time $O(n^{m/\varepsilon^3}+mn^8)$. 


{\bf Stage 1:}
 Intuitively, this is a pre-processing stage with the purpose of picking a small initial set consisting of elements with 'large' marginal values, i.e. marginal value at least $\varepsilon^3 V_i$ for some function $f_i$. This is necessary for technical reasons due to the rounding procedure in Stage 3. 

Given a set $S$ of size $k$, fix a function $f_i$ and index elements in $S=\{s_1,\dots,s_k\}$ in the order in which the greedy algorithm would pick them. There are at most $1/\varepsilon^3$ elements such that $f_i(s_j|\{s_1,\dots,s_{j-1}\})\geq \varepsilon^3V_i$, since otherwise by monotonicity $f_i(S)>V_i$ (violating our w.l.o.g.\ assumption that $f_i(S)\leq V_i\, \forall i$). In fact, due to decreasing marginal values we have, $f_i(s_j|\{s_1,\dots,s_{j-1}\})<\varepsilon^3V_i$ for every $j>1/\varepsilon^3$.

Therefore, we focus on sets of size $\leq m/\varepsilon^3$ (at most $1/\varepsilon^3$ elements for each function) to find an initial set such that the remaining elements have marginal value $\leq \varepsilon^3 V_i$ for $f_i$, for every $i$. In particular, one can try all possible initial sets of this size (i.e. run subsequent stages with different starting sets), leading to the $n^{m/\varepsilon^3}$ term in the runtime. Stages 2 and 3 have runtime polynomial in $m$ (in fact Stage $3$ has runtime independent of $m$). Hence, Stage 1 is really the bottleneck. For the more general case of matroid constraint, it is not obvious at all if one can do better than brute force enumeration over all possible starting sets and still retain the approximation guarantee. However, we will show that for cardinality constraints one can easily avoid enumeration.

{\bf Stage 2:} 
Given a starting set $S$ from stage one, this stage works with the ground set $N-S$ and runs the continuous greedy algorithm. Suppose a feasible set $S_k$ exists for the problem, then for the right starting set $S_1\subset S_k$, this stage outputs a fractional point ${\bf x}(k_1)\in[0,1]^n$ with $|{\bf x}(k_1)|=k_1=k-|S_1|$ such that $F_i({\bf x}(k_1)|{\bf x}_{S_1})\geq (1-1/e) (V_i-f_i(S_1))$ for every $i$. 
However, this is computationally expensive and takes time $\tilde{O}(mn^8)$. We formally summarize this stage in the following lemma and refer the interested reader to \cite{swap} for further details (which will not be necessary for subsequent discussion).
\begin{lem}\label{congreedy}
	(\cite{swap}, Lemma 7.3) Given submodular functions $f_i$ and values $V_i$, cardinality $k$, the continuous greedy algorithm finds a point ${\bf x}\in[0,1]^{n}$ such that 
	$F_i({\bf x}(k))\geq (1-1/e)V_i, \forall i$ or outputs a certificate of infeasibility.
\end{lem}

{\bf Stage 3:}  For the right starting set $S_1$ (if one exists), Stage 2 successfully outputs a point ${\bf x}(k_1)$. 
Stage 3 now follows a random process that converts ${\bf x}(k_1)$ into a set $S_2$ of size $k_1$ such that, $S_2\in N-S_1$ and $f_i(S_1\cup S_2)\geq (1-1/e)(1-\varepsilon)V_i, \forall i$ as long as $\varepsilon$ is ``small enough". 
The rounding procedure developed in \cite{swap}, called \emph{swap rounding}, is designed to work more generally for matroid constraints. For the special case of uniform matroid that we focus on, it turns out that \emph{independently rounding} ${\bf x}(k_1)$ after an appropriate scaling suffices\footnote{An earlier version of this work \cite{neurips}, uses the more involved swap rounding technique instead.} i.e., we sample element $i$ independently w.p.\ $\eta x_i(k_1)$ for a suitably chosen $\eta>0$. In order to formalize this we introduce the following notation,
\[R({\bf x}):=\text{ random set obtained by independently sampling elements according to ${\bf x}$.}\]
Since any element $i$ is part of $R({\bf x})$ w.p.\ $x_i$, we naturally have $\mathbb{E}[f(R({\bf x}))]=F({\bf x})$. To formally show that independently rounding suffices for cardinality constraints, we need the following lemmas.
\begin{lem}\label{chern} (\cite{swap}, Theorem 1.3) 
	 Let $f$ be a monotone submodular function with the maximum function value of singletons in $[0,1]$. For any ${\bf x}\in[0,1]^n$. 
	and $\delta>0$,
	\[ \Pr[f(R({\bf x}))\leq (1-\delta)F({\bf x})] \leq e^{-F({\bf x})\delta^2/2}. \]	
\end{lem}
\begin{lem}\label{concavity}
	Given a point ${\bf x}\in[0,1]^n$ with $|{\bf x}|=k$ and a multilinear extension $F$ of a monotone submodular function, for every $k_1\leq k$, 
	$$F\Big(\frac{k_1}{k}{\bf x}\Big)\geq \frac{k_1}{k}F({\bf x}).$$
\end{lem}
\begin{proof}{Proof}
	Note that the statement is true for concave functions. Since multilinear extensions are concave in positive directions (Section 2.1 of \cite{vond1}), we have the desired.
\end{proof}

Using these lemmas we now formalize the result on rounding below.

\begin{lem}\label{swap}
{\color{black}	
	 Suppose we are given $m\geq 2$ non-negative values $\{T_i\}_{i\in[m]}$, monotone submodular functions $f_i(.)$ with the maximum value of singletons in $[0,\varepsilon^3T_i]$ for every function, and a fractional point ${\bf x}$ such that $F_i({\bf x})\geq (1-1/e) T_i$ for every $i\in[m]$. }Let $|{\bf x}|\leq k$, then for $\eta=\Big(1-\sqrt{\frac{\ln k}{k}}\Big)$ and 
	  $\varepsilon\leq \frac{1}{10 \ln m}$, we have,
	\[ \Pr\Big[\{|R(\eta {\bf x})|>k\} \vee \{\exists\ i\in[m]\ s.t.\ f_i(R(\eta{\bf x}))<(1-\varepsilon)\eta F_i({\bf x})\}\Big]< \frac{1}{k^{O(1)}}+ \frac{1}{m^{O(1)}}.\]
\end{lem}
\begin{proof}{Proof}
	 First, from Lemma \ref{concavity} we have that $F(\eta {\bf x})\geq \eta F({\bf x})$. Further, for every $i\in [m]$, function $f_i$ has singleton values in $[0,\varepsilon^3T_i]$ so using Lemma \ref{chern} for the scaled function $\frac{1}{\varepsilon^3 T_i} f_i$ and its multilinear extension $\frac{1}{\varepsilon^3 T_i} F_i$,  we have that,
	\[ \Pr[f_i(R(\eta {\bf x}))\leq(1-\varepsilon)\eta F({\bf x})]\leq e^{-\frac{\eta F({\bf x})}{2\varepsilon T_i}}< e^{-\eta\frac{1-1/e}{2\varepsilon}},\]
	where the last inequality uses the fact that $F({\bf x})> (1-1/e) T_i$. 
	Additionally, from standard Chernoff bound for Bernoulli r.v.s, we have that,
	\[\Pr\big[|R(\eta {\bf x})|\geq k\big]\leq \Pr\Big[|R(\eta {\bf x})|\geq \Big(1+\sqrt{\frac{\ln k}{k}}\Big)\eta k\Big]\leq e^{-\frac{\delta^2}{2+\delta}k}< e^{-\frac{\ln k}{3}}. \] The rest follows via union bound, noting that $5\eta(1-1/e)-1\geq 0$ for all $k$.
\end{proof}
\emph{Remark:} 
The above can be converted to a result w.h.p.\ by standard repetition. 
\subsection{Some simple heuristics}\label{standheu}
Before we present the main results, let us take a step back and examine some variants of the standard greedy algorithm. To design a greedy heuristic for multiple functions, what should the objective for greedy selection be?

One possibility is to split the selection of $k$ elements into $m$ equal parts. In part $i$, pick $k/m$ elements greedily w.r.t. function $f_i$. It is not difficult to see that this is a (tight) $\beta(k/m)$ approximation. Second, recall that the convex combination of monotone submodular functions is also monotone and submodular. Therefore, one could run the greedy algorithm on a fixed convex combination of the $m$ functions. It can be shown this does not lead to an approximation better than $1/\Theta(m)$. In fact, this is the idea behind the bi-criterion approximation in \cite{main}.   
Third, one could select elements greedily w.r.t. to the objective function $h(.)=\min_i f_i(.)$. A na\"ive implementation of this algorithm can have arbitrarily bad performance even for $m=2$ (previously observed in \cite{ipcofull}). We show later in Section \ref{resolve}, that if one greedily picks sets of size $k'$ instead of singletons at each step, for large enough $k'$ one can get arbitrarily close to $(1-1/e)$.
\section{Main Results}

%

%

\subsection{Asymptotic $(1-1/e)$ approximation for $m=o\big(\frac{k}{\log^3 {k}}\big)$}\label{step1}
We replace the enumeration in Stage 1 with a single starting set, obtained by scanning once over the ground set. The main idea is simply that for the cardinality constraint case, any starting set that fulfills the Stage 3 requirement of small marginals will be acceptable (not true for general matroids).

Before proceeding, recall that in the parameter regime $k<m$, the problem is in general inapproximable due to a reduction from the hitting set problem given by \cite{main}. Ideally, we would then like to design an approximation algorithm for the regime $k\geq m$. In our algorithm and results below, we require the slightly stronger assumption that $k > m/\varepsilon^3$. Given that we choose $\varepsilon=\min\{\frac{1}{10\ln m},\sqrt[4]{\frac{m}{k}}\}$, this translates to, \[k>\Omega(m\ln^3 m).\] 
For smaller $k$ our algorithm does not have a constant factor guarantee as we approach the inapproximability regime.

{\bf New Stage 1:}
Start with $S_1=\emptyset$ and pass over all elements once in an arbitrary order. 
{\color{black}\emph{For each element $e$, add it to $S_1$ if there exists an $i\in[m]$ such that $f_i(e|S_1)\geq \varepsilon^3\left(V_i-f_i(S_1)\right)$ and $f_i(S_1)<(1-1/e)V_i$.} When all elements have been parsed and the process terminates, let $M_1$ be the set of functions $f_i$ such that $f_i(S_1)\geq (1-1/e)V_i$.
	\begin{lem}\label{needed}
		The new stage 1 outputs a set $S_1$ of size at most $m/\varepsilon^3$. On termination, for every element $e\in N\backslash S_1$ and function $f_i\not \in M_1$ we have, $f_i(e|S_1)<\varepsilon^3\left(V_i-f_i(S_1)\right)$. 
	\end{lem}
The proof of Lemma \ref{needed} is included in Appendix \ref{appx:need}. The first half of the lemma follows by mimicking the standard analysis of greedy algorithms for submodular maximization. The second half follows by definition of new stage 1. We note that peer-reviewed versions of this paper have a different subroutine in stage 1. The current version presents a new subroutine in stage 1 and fixes an error present in all the earlier versions that was discovered after publication. 

Notice that for all functions in $M_1$ we have achieved the desired target. So from here on we ignore functions in $M_1$ and perform the remaining stages only on functions $f_i\not\in M_1$.} 

Let $k_1=k-|S_1|$ and note $k_1\geq k- m/\varepsilon^3$.
Stage 2 remains the same as Algorithm 1 and outputs a fractional point ${\bf x}(k_1)$ with $|{\bf x}(k_1)|=k_1$. While enumeration over all starting sets allowed us to find a starting set such that $F_i({\bf x}(k_1)|{\bf x}_{S_1})\geq (1-1/e) (V_i-f_i(S_1))$ for every $i$; with the new Stage 1 we will use Lemma \ref{concavity} to get the desired 
 lower bound on the marginal value of ${\bf x}(k_1)$. 

\begin{lem}\label{stage2}
	$F_i({\bf x}(k_1)|{\bf x}_{S_1})\geq \beta(1)\frac{k_1}{k} (V_i-f_i(S_1))$ for every $f_i\not\in M_1$. 
\end{lem}
\begin{proof}{Proof}
	Recall that $S_k$ denotes a feasible solution with cardinality $k$, and let ${\bf x}_{S_k}$ denote its characteristic vector. Clearly, $|{\bf x}_{S_k \backslash S_1}|\leq k$ and $F_i({\bf x}_{S_k \backslash S_1}|{\bf x}_{S_1})=f_i(S_k | S_1)\geq (V_i-f_i(S_1))$ for very $i$. And now from Lemma \ref{concavity}, we have that there exists a point ${\bf x}'$ with $|{\bf x}'|= k_1$ such that $F_i({\bf x}'|{\bf x}_{S_1})\geq\frac{k_1}{k} F_i({\bf x}_{S_k \backslash S_1}|{\bf x}_{S_1})$ for every $i$. Finally, using Lemma \ref{congreedy} we have $F_i({\bf x}(k_1)|{\bf x}_{S_1})\geq \beta(1)F_i({\bf x}'|S_1)$, which gives the desired bound. 
\end{proof}
Stage 3 rounds ${\bf x}(k_1)$ to $S_2$ of size $k_1$, and final output is $S_1\cup S_2$. The following theorem now completes the analysis.

\begin{thm}\label{extension}
	For $\varepsilon=\min\{\frac{1}{10\ln m},\sqrt[4]{\frac{m}{k}}\}$ and $\eta=1-\sqrt{{\log k_1/k_1}}$ we have, 
	\[f_i(S_1\cup S_2)\geq (1-\varepsilon)(1-m/k\varepsilon^3)\eta \beta(1)V_i\, \forall i,\] with constant probability. For $m=o\big(k/\log^3 {k}\big)$, the multiplicative factor is asymptotically $(1-1/e)$.
\end{thm}

%

\begin{proof}{Proof}
{\color{black}	Using Lemma \ref{needed} and Lemma \ref{stage2}, we apply Lemma \ref{swap} with values $T_i= V_i-f_i(S_1)$ for every $f_i\not\in M_1$. Consequently, $f_i(S_2|S_1)\geq (1-\varepsilon) (1-m/k\varepsilon^3) \eta \beta(1)(V_i-f_i(S_1))$ for every $f_i\not\in M_1$.} 
	Therefore, $f_i(S_1\cup S_2)\geq
	(1-\varepsilon)(1-m/k\varepsilon^3)\eta \beta(1) V_i$ for every $i\in[m]$. To refine the guarantee, we choose $\varepsilon=\min\{\frac{1}{10\ln m},\sqrt[4]{\frac{m}{k}}\}$, where the $\frac{1}{10\ln m}$ is due to Lemma \ref{swap} and the $\sqrt[4]{\frac{m}{k}}$ term is to balance $\varepsilon$ and $m/k\varepsilon^3$. The resulting guarantee becomes $(1-1/e)(1-h(m,k))$, where the function $h(m,k)\to 0$ as $k\to \infty$, so long as $m=o\big(\frac{k}{\log^3 {k}}\big)$. 
	
	 The first stage makes $O(mn)$ oracle queries, the second stage runs the continuous greedy algorithm on all functions simultaneously and makes $\tilde{O}(n^8)$ queries to each function oracle, contributing $O(mn^8)$ to the runtime. 
	 Finally, the simple independent rounding in Stage 3 takes time $O(n)$. 
\end{proof}

\noindent \textbf{Curvature optimal result:} Recall, the curvature of a submodular function  $f$ is a parameter $c\in[0,1]$ such that for any set $S\subset N$ and element $i\not\in S$, we have $f(S\cup \{j\})-f(S)\geq (1-c)f(\{j\}).$ Further, given multiple functions $f_i$, we let $c$ denote the maximum curvature $\max_{i}c_i$.

Now, from Lemma 3.1 and Theorem 3.2 in \cite{curve}, the multiplicative factor $\beta(1)$ in Lemma \ref{stage2} generalizes to $\frac{1}{c}\beta(c)$. Consequently, the asymptotic guarantee in Theorem \ref{extension} is given by $\frac{1}{c}(1-e^{-c})$ for maximum curvature $c$. This for instance, implies an asymptotic guarantee of 1 for linear functions. Furthermore, due to Theorem 4.1 in \cite{curve}, this is the optimal dependence on the curvature parameter.

\subsection{Fast, asymptotic $(1-1/e)^2-\delta$ approximation for $m=o\big(\frac{k}{\log^3 {k}}\big)$}\label{keyfast}
While the previous algorithm achieves the best possible asymptotic guarantee, it is infeasible to use in practice. The main underlying issue was our usage of the continuous greedy algorithm in Stage 2 which has runtime $\tilde{O}(mn^8)$, but the flexibility offered by continuous greedy was key to maximizing the multilinear extensions of all functions at once. 
To improve the runtime we avoid continuous greedy and find an alternative 
in Multiplicative-Weight-Updates (MWU) instead. MWU allows us to combine multiple submodular objectives together into a single submodular objective and utilize fast algorithms for $P_0$ at every step.

The algorithm consists of 3 stages as before. Stage 1 remains the same as the New Stage 1 introduced in the previous section. Let $S_1$ be the output of this stage as before.
Stage 2 is replaced with a fast MWU based subroutine that runs for $T=O(\frac{\ln m}{\delta^2})$ rounds and solves an instance of $SO$ during each round. Here $\delta$ is an artifact of MWU and manifests as a subtractive term in the approximation guarantee. The nearly linear time algorithm for $SO$, in \cite{lazyfast}, has runtime $O(n\log \frac{1}{\delta'})$ and an \emph{expected} guarantee of $(1-1/e)-\delta'$. The slightly slower, but still nearly linear time $O(\frac{n}{\delta'}\log \frac{n}{\delta'})$ \emph{thresholding} algorithm in \cite{fast}, has (the usual) deterministic guarantee of $(1-1/e)-\delta'$. Both of these are known to perform well in practice and using either would lead to a runtime of $T\times \tilde{O}(n/\delta)=\tilde{O}(\frac{n}{\delta^3})$, which is a vast improvement over the previous algorithm. 

Now, fix some algorithm $\mathcal{A}$ for $P_0$ with guarantee $\alpha$, and let $\mathcal{A}(f,k)$ denote the set it outputs given monotone submodular function $f$ and cardinality constraint $k$ as input. Note that $\alpha$ can be as large as $1-1/e$, and we have $k_1=k-|S_1|$ as before. Then the new Stage 2 is, \\
\setcounter{algorithm}{1}
\begin{algorithm}[H]
	\caption{Stage 2: MWU}
	\label{mwu2}
	\begin{algorithmic}[1]
		\State {{\bf Input:} $\delta, T=\frac{2 \ln m}{\delta^2},\lambda^1_i=1/m,\tilde{f}_i(.)=\frac{f_i(.|S_1)}{V_i-f_i(S_1)}$}
		\While{$1\leq t\leq  T$} 
		$g^t(.)=\sum_{i=1}^{m} \lambda^t_i \tilde{f}_i(.)$\\
		$ X^t=\mathcal{A}(g^t,k_1)$ \\
		$m^t_i=\tilde{f}_i(X^t)-\alpha$\\
		$\lambda^{t+1}_i=\lambda^t_i (1-\delta m^t_i)$\\
		$t=t+1$
		\EndWhile
		\State {{\bf Output:} ${\bf x}_2=\frac{1}{T}\sum_{t=1}^T X^t$}
	\end{algorithmic}
\end{algorithm}

\emph{The point ${\bf x}_2$ obtained above is rounded to a set $S_2$ in Stage 3 (which remains unchanged). The final output is $S_1\cup S_2$}. Note that with abuse of notation we used the sets $X^t$ to also denote the respective support vectors. We continue to use $X^t$ and ${\bf x}_{X^t}$ interchangeably in the below.

This application of MWU is unlike \cite{azar,chekurimwu}, where broadly speaking, the framework is applied in a novel way to determine how an individual element is picked or how a direction for movement is chosen in case of continuous greedy. In contrast, we use standard algorithms for $P_0$ and pick an entire set before changing weights. Also, \cite{chekurimwu} use MWU along with the continuous greedy framework whereas, we use MWU to replace the continuous greedy framework. Subsequent to our work we discovered a resembling application of MWU in \cite{chen}. Their application differs from Algorithm \ref{mwu2} only in minor details, but unlike our result they give a bi-criterion approximation where the output is a set $S$ of cardinality up to $k\frac{\log m}{V\varepsilon^2}$ such that $f_i(S)\geq(1-1/e-2\varepsilon) V$. 

Now, consider the following intuitive schema. We would like to find a set $X$ of size $k$ such that $f_i(X)\geq \alpha V_i$ for every $i$. While this seems hard, consider the combination $\sum_i \lambda_i f_i(.)$, which is also monotone submodular for non-negative $\lambda_i$. We can easily find a set $X_{\lambda}$ such that $\sum_i \lambda_i f_i(X_{\lambda})\geq \sum_i \lambda_i V_i$, since this is a single objective problem and we have fast approximations for $P_0$. However, for a fixed set of scalar weights $\lambda_i$, solving the $P_0$ problem instance need not give a set that has sufficient value for every individual function $f_i(.)$. This is where MWU comes into the picture. We start with uniform weights for functions, solve an instance of $P_0$ to get a set $X^1$. Then we change weights to undermine the functions for which $f_i(X^1)$ was closer to the target value and stress more on functions for which $f_i(X^1)$ was small, and repeat now with new weights. After running many rounds of this, we have a collection of sets $X^t$ for $t\in\{1,\dots,T\}$. Using tricks from standard MWU analysis (\cite{arora}) along with submodularity and monotonicity, we show that $\sum_t \frac{f_i(X^t|S_1)}{T}\gtrapprox (1-1/e)(V_i-f_i(S_1))$. Thus far, this resembles how MWU has been used in the literature for linear objectives, for instance the Plotkin-Shmoys-Tardos framework for solving LPs. However, a new issue now arises due to the non-linearity of functions $f_i$. As an example, suppose that by some coincidence ${\bf x}_2=\frac{1}{T}\sum_{t=1}^T X^t$ turns out to be a binary vector, so we easily obtain the set $S_2$ from ${\bf x}_2$. We want to lower bound $f_i(S_2|S_1)$, and while we have a good lower bound on $\sum_t \frac{f_i(X^t|S_1)}{T}$, it is unclear how the two quantities are related. More generally, we would like to show that $F_i({\bf x}_2|{\bf x}_{S_1})\geq \beta \sum_t \frac{f_i(X^t|S_1)}{T}$ and this would then give us a $\beta\alpha=\beta(1-1/e)$ approximation using Lemma \ref{swap}. Indeed, we show that $\beta\geq (1-1/e)$, resulting in a $(1-1/e)^2$ approximation. Now, we state and prove lemmas that formalize the above intuition. 


\begin{lem}\label{feasible}
	$g^t(X^t)\geq \frac{k_1}{k}\alpha\sum_i\lambda^t_i	, \forall t$. 
\end{lem}
\begin{proof}{Proof}
	 Consider the optimal set $S_k$ and note that $\sum_i \lambda^t_i \tilde{f}_i(S_k)\geq \sum_i \lambda^t_i, \forall t$. Now the function $g^t(.)=\sum_i \lambda^t_i \tilde{f}_i(.)$, being a convex combination of monotone submodular functions, is also monotone submodular. We would like to show that there exists a set $S'$ of size $k_1$ such that $g^t(S')\geq \frac{k_1}{k}\sum_i\lambda^t_i$. Then the claim follows from the fact that $\mathcal{A}$ is an $\alpha$ approximation for monotone submodular maximization with cardinality constraint. 
	
	To see the existence of such a set $S'$, greedily index the elements of $S_k$ using $g^t(.)$. Suppose that the resulting order is $\{s_1,\dots,s_k\}$, where $s_i$ is such that $g^t(s_i|\{s_1,\dots,s_{i-1}\})\geq g^t(s_j|\{s_1,\dots,s_{i-1}\})$ for every $j>i$. Then the truncated set $\{s_1,\dots,s_{k-|S_1|}\}$ has the desired property, and we are done. 
\end{proof}
%
\begin{lem}\label{frommwu}
	$$	\frac{\sum_t \tilde{f}_i(X^t)}{T}\geq \frac{k_1}{k}\alpha-\delta\, , \forall i.$$
\end{lem}  
\begin{proof}{Proof}
	Suppose we have,
	\begin{equation}
	\frac{\sum_t (\tilde{f}_i(X^t)-\alpha)}{T} +\delta \geq \frac{1}{T}\sum_t \sum_i \frac{\lambda^t_i}{\sum_i \lambda^t_i} (\tilde{f}_i(X^t)-\alpha),\forall i. \label{step}
	\end{equation}
	Using Lemma \ref{feasible} the RHS simplifies to,
	\begin{eqnarray*}
		\frac{1}{T} \sum_t \frac{g(X^t)}{\sum_i \lambda^t_i}-\alpha &&\geq \alpha\Big(\frac{k_1}{k}-1\Big), 
	\end{eqnarray*}
	and for every $i$,
	\begin{eqnarray*}
		\frac{\sum_t (\tilde{f}_i(X^t)-\alpha)}{T}+\delta &&\geq \alpha \Big(\frac{k_1}{k}-1\Big)\\
		\frac{\sum_t \tilde{f}_i(X^t)}{T}&&\geq \frac{k_1}{k}\alpha-\delta.
	\end{eqnarray*}
	To finish the proof we need to show (\ref{step}), and this closely resembles the analysis in Theorem 3.3 and 2.1 in \cite{arora}. We will use the potential function $\Phi^t=\sum_i \lambda^t_i$. Let $p^t_i=\lambda^t_i/\Phi^t$ and $M^t=\sum_i p^t_i m^t_i$. Then we have,
	\begin{eqnarray*}
		\Phi^{t+1}&&=\sum_i \lambda^t_i (1-\delta m^t_i)\\
		&&=\Phi^t - \delta \Phi^t\sum_i p^t_i m^t_i\\
		&&=\Phi^t(1-\delta M^t) \leq \Phi^t e^{-\delta M^t}
	\end{eqnarray*}
	After $T$ rounds, $\Phi^T\leq \Phi^1 e^{-\delta\sum_t M^t}$. Further, for every $i$,
	\begin{eqnarray*}
		&\Phi^T\geq w^T_i =\frac{1}{m} \prod_t (1-\delta m^t_i)\\
		&\ln(\Phi^1 e^{-\delta\sum_t M^t})\geq \sum_t \ln (1-\delta m^t_i)-\ln m\\
		& \delta\sum_t M^t \leq \ln m + \sum_t \ln (1-\delta m^t_i)
	\end{eqnarray*}
	Using $\ln(\frac{1}{1-\varepsilon})\leq \varepsilon+\varepsilon^2$ and $\ln (1+\varepsilon)\geq \varepsilon-\varepsilon^2$ for $\varepsilon\leq 0.5$, and with $T=\frac{2\ln m}{\delta^2}$ and $\delta< (1-1/e)$ (for a positive approximation guarantee), we have, $$\frac{\sum_t M^t}{T}\leq \delta + \frac{\sum_t m^t_i}{T},\forall i. $$
	
\end{proof}

\begin{lem}\label{key}
	Given monotone submodular function $f$, its multilinear extension $F$, sets $X^t$ for $t\in\{1,\dots,T\}$, and a point ${\bf x}=\sum_t X^t /T $, we have,	$$F({\bf x}) \geq (1-1/e)\frac{1}{T}\sum_{t=1}^T f(X^t).$$
\end{lem}  
\begin{proof}{Proof}
Consider the concave closure of a submodular function $f$, 
$$f^+({\bf x})=\max_{{\bf \alpha}}\{\sum_{X} \alpha_Xf(X)|\sum_{X} \alpha_X X={\bf x}, \sum_{X} \alpha_X\leq 1, \alpha_X\geq 0\, \forall X\subseteq N\}.$$ Clearly, $f_i^+({\bf x})\geq \frac{\sum_t f_i(X^t)}{T}$.  So it suffices to show $F_i({\bf x}) \geq (1-1/e)f^+_i({\bf x})$, which in fact, follows from Lemmas 4 and 5 in \cite{calines}. A novel and direct proof for this statement is included in Appendix \ref{keyred}.

\end{proof}

\begin{thm}
	For $\varepsilon=\min\{\frac{1}{20\ln m},\sqrt[4]{\frac{m}{k}}\}$, the algorithm makes $O(\frac{n}{\delta^3}\log m\log \frac{n}{\delta})$ queries, and with constant probability outputs a feasible $(1-\varepsilon)(1-\frac{m}{k\varepsilon^3})(1-1/e)^2-\delta$ approximate set. Asymptotically, $(1-1/e)^2-\delta$ approximate for $m=o\big(k/\log^3 {k}\big)$.
\end{thm}
\begin{proof}{Proof}
	%
	 Using the thresholding algorithm in \cite{fast} as the subroutine $\mathcal{A}$, we combine Lemmas \ref{frommwu} \& \ref{key} with $\alpha=(1-1/e)-\delta$ to get, \[\tilde{F}_i({\bf x}_2)\geq(1-1/e) \frac{\sum_t \tilde{f}_i(X^t)}{T}\geq \frac{k_1}{k}(1-1/e)^2-2\delta\, , \forall i.\] The asymptotic result follows just as in Theorem \ref{extension}. 
	%
	For runtime, note that Stage 1 takes time $O(n)$. Stage 2 runs an instance of $\mathcal{A}(.)$, $T$ times, leading to an upper bound of $O((\frac{n}{\delta}\log \frac{n}{\delta})\times \frac{\log m}{\delta^2})=O(\frac{n}{\delta^3}\log m\log \frac{n}{\delta})$. 
	Finally, 
	rounding takes $O(n)$ time.  
	Combining all three we get a runtime of $O(\frac{n}{\delta^3}\log m\log \frac{n}{\delta})$.
\end{proof}

\subsection{Variation in optimal solution value and derandomization}\label{resolve}
Consider the problem $P_0$ with cardinality constraint $k$. Given an optimal solution $S_k$ with value $OPT_k$ for the problem, it is not difficult to see that for arbitrary $k'\leq k$, there is a subset $S_{k'}\subseteq S_k$ of size $k'$, such that $f(S_{k'})\geq \frac{k'}{k} OPT_k$. For instance, 
indexing the elements in $S_k$ using the greedy algorithm, and choosing the set given by the first $k'$ elements gives such a set. This implies $OPT_{k'}\geq \frac{k'}{k} OPT_k$, and the bound is easily seen to be tight.

This raises a natural question: Can we generalize this bound on variation of optimal solution value with varying $k$, for multi-objective maximization? A priori, this isn't completely obvious (to us) even for modular functions. In particular, note that indexing elements in order they are picked by the greedy algorithm doesn't suffice since there are many functions and we need to balance values amongst all. We show below that one can indeed derive such a bound. 

\begin{lem}\label{variation}
	Given that there exists a set $S_k$ such that $f_i(S_k)\geq V_i,\forall i$, let $\varepsilon\leq \frac{1}{16\ln m}$. For every $k'\in[m/\varepsilon^3 , k]$, there exists $S_{k'}\subseteq S_k$ of size $k'$, such that, $$f_i(S_{k'})\geq(1-\varepsilon)\Big(\frac{k'-m/\varepsilon^3}{k-m/\varepsilon^3} \Big)V_i,\forall i .$$
\end{lem}
\begin{proof}{Proof}
	We restrict our ground set of elements to $S_k$ and let $S_1$ be a subset of size at most $m/\varepsilon^3$ such that $f_i(e|S_1)<\varepsilon^3 V_i, \forall e\in S_k\backslash S_1 \text{ and }\forall i$ (recall, we discussed the existence of such a set in Section 2.1, Stage 1). The rest of the proof is similar to the proof of Lemma \ref{stage2}.
	Consider the point ${\bf x}=\frac{k'-|S_1|}{k-|S_1|}{\bf x}_{S_k\backslash S_1}$. Clearly, $|{\bf x}|=k'-|S_1|$, and from Lemma \ref{concavity} we have, \[F_i({\bf x}|{\bf x}_{S_1})\geq\frac{k'-|S_1|}{k-|S_1|}F_i({\bf x}_{S_k\backslash S_1}|{\bf x}_{S_1})=\frac{k'-|S_1|}{k-|S_1|}f_i(S_k\backslash S_1|S_1)\geq  \frac{k'-|S_1|}{k-|S_1|}(V_i-f_i(S_1)),\forall i.\]
	 Finally, using the concentration result for swap rounding (see Theorem 1.4 in \cite{swap}), we have the existence of a set $S_2$ of size at most $k'-|S_1|$, such that $f_i(S_1\cup S_2)\geq (1-\varepsilon)\frac{k'-|S_1|}{k-|S_1|}V_i,\forall i$. Note that using the concentration inequality due to swap rounding allows us to eliminate the additional factor $\eta$ that we get in independent rounding (see Lemma \ref{swap}), at the cost of choosing a smaller $\varepsilon$.
	
\end{proof}
{\bf Conjecture in \cite{ipcofull}}: Note that this resolves a slightly weaker version of the conjecture in \cite{ipcofull} for constant $m$. 
The original conjecture states that for constant $m$ and every $k'\geq m$, there exists a set $S$ of size $k'$, such that $f_i(S)\geq \frac{k'-\Theta(1)}{k} V_i,\forall i$. Asymptotically, both $\frac{k'-m/\varepsilon^3}{k-m/\varepsilon^3}$ and $\frac{k'-\Theta(1)}{k}$ tend to $\frac{k'}{k}$. This implies that for large enough $k'$, we can choose sets of size $k'$ ($k'$-tuples) at each step to get a deterministic (asymptotically) $(1-1/e)-\varepsilon$ approximation with runtime $O(kn^{m/\varepsilon^4})$ for the multi-objective maximization problem, when $m$ is constant (all previously known approximation algorithms, as well as the ones presented earlier, are randomized).
\begin{thm}
	For $k'=\frac{m}{\varepsilon^4}$, choosing $k'$-tuples greedily w.r.t.\ $h(.)=\min_i f_i(.)$ yields approximation guarantee $(1-1/e)(1-2\varepsilon)$ for $k\to \infty$, while making $n^{m/\varepsilon^4}$ queries.
\end{thm}
\begin{proof}{Proof}
	The analysis generalizes that of the standard greedy algorithm (\cite{nem1,nem}). Let $S_j$ denote the set at the end of iteration $j$. $S_0=\emptyset$ and let the final set be $S_{\floor*{k/k'}}$. Then from Theorem \ref{variation}, we have that at step $j+1$, there is some set $X\in S_k\backslash S_{j}$ of size $k'$ such that 
	$$f_i(X|S_{j})\geq (1-\varepsilon)\frac{k'-m/\varepsilon^3}{k-m/\varepsilon^3}\big(V_i-f_i(S_j)\big), \forall i.$$ 
	To simplify presentation let $\eta=(1-\varepsilon)\frac{k'-m/\varepsilon^3}{k-m/\varepsilon^3}$ and note that $\eta\leq 1$. Further, $1/\eta\to \infty$ as $k\to \infty$ for fixed $m$ and $k'=o(k)$. Now, we have for every $i$, $f_i(S_{j+1})-(1- \eta)f_i(S_j)\geq \eta V_i$. Call this inequality $j+1$. Observe that inequality $\floor*{k/k'}$ states $f_i(S_{\floor*{k/k'}})-(1-\eta)f_i(S_{\floor*{k/k'}-1})\geq \eta V_i,\forall i$. Therefore, multiplying inequality $\floor*{k/k'}-j$ by $(1-\eta)^j$ and telescoping over $j$ we get for every $i$,
	\begin{eqnarray*}
		f_i(S_{\floor*{k/k'}})&&\geq \sum_{j=0}^{\floor*{k/k'}-1} (1-\eta)^j \eta V_i\\
		&&\geq \big(1-(1-\eta)^{\floor*{k/k'}})V_i\\
		&&\geq  \big(1-(1-\eta)^{\frac{1}{\eta} \eta\floor*{k/k'}})V_i\\
		&&\geq \beta(\eta \floor*{k/k'})V_i\quad  \geq (1-1/e)(\eta \floor*{k/k'})V_i.
	\end{eqnarray*}
	Where we used (\ref{simfact}) for the last inequality. Let $\varepsilon=\sqrt[4]{\frac{m}{k'}}$, then we have,
	\begin{eqnarray*}
		\eta \floor*{k/k'}\geq (1-\varepsilon)\frac{1-m/k'\varepsilon^3}{1-m/k\varepsilon^3} \Big(1-\frac{k'}{k}\Big)\geq\frac{\Big(1-\sqrt[4]{\frac{m}{k'}}\Big)^2}{1-\frac{1}{k}\sqrt[4]{\frac{m}{(k')^3}}}\Big(1-\frac{k'}{k}\Big)
	\end{eqnarray*} 
	As $k\to \infty$ we get the asymptotic guarantee $(1-1/e)\Big(1-\sqrt[4]{\frac{m}{k'}}\Big)^2=(1-1/e)(1-\varepsilon)^2$.
\end{proof}
\section{Experiments on Kronecker Graphs}

We choose synthetic experiments where we can control the parameters to see how the algorithm performs in various scenarios, esp.\ since we would like to test how the MWU algorithm performs for small values of $k$ and $m=\Omega(k)$. 
We work with formulation $P_1$ of the problem and consider a multi-objective version of the \textit{max-k-cover} problem on graphs. 
Random graphs for our experiments were generated using the Kronecker graph framework introduced in \cite{kron}. These graphs exhibit several natural properties and are considered a good approximation for real networks (esp. social networks \cite{influ}).  

We compare three algorithms: (i) A baseline greedy heuristic, labeled GREEDY, which focuses on one objective at a time and successively picks $k/m$ elements greedily w.r.t.\ each function (formally stated below). (ii) A bi-criterion approximation called SATURATE from \cite{main}, to the best of our knowledge this is considered state-of-the-art for the problem. (iii) We compare these algorithms to a heuristic inspired by our MWU algorithm. This heuristic differs from the algorithm discussed earlier in two ways. Firstly, we eliminate Stage 1 which was key for technical analysis but in practice makes the algorithm perform similar to GREEDY. Second, instead of simply using the the rounded set $S_2$, we output the best set out of $\{X^1,\dots,X^T\}$ and $S_2$. Also, for both SATURATE and MWU we estimate target value $t$ using binary search and consider capped functions $\min \{f_i(.),t\}$. Also, for the MWU stage, we used $\delta=0.5$ or $0.2$.

\setcounter{algorithm}{2}
\begin{algorithm}[H]
	\caption{GREEDY}
	
	\begin{algorithmic}[1]
		\State {{\bf Input:} $k,m,f_i(.)\text{ for } i\in[m]$}
		\State $S=\emptyset, i=1$
		\While{$|S|\leq k-1$}\\ 
		$S=S+ \arg\max_{x\in N-S} f_i(x|S)$\\
		$i= i+1 \mod m+1$
		\EndWhile
		\State {{\bf Output:} $S$}
	\end{algorithmic}
\end{algorithm} 
		\begin{algorithm}[H]
			\caption{SATURATE}\label{saturate}
			\begin{algorithmic}[1]
				\State Input: $k,t,f_1,\dots,f_m$ and set $A=\emptyset$
				\State $g(.)= \sum_i \min\{f_i(.),t\}$
				\While{$|A|< k$} $A=A+\argmax{x\in N-A} g(x|A)$ \EndWhile
				\vspace{-3mm}
				\State Output: $A$
			\end{algorithmic}
		\end{algorithm}
We pick Kronecker graphs of sizes $n\in \{64,512,1024\}$ with random initiator matrix \footnote{To generate a Kronecker graph one needs a small initiator matrix. Using \cite{kron} as a guideline we use random matrices of size $2\times2$, each entry chosen uniformly randomly (and independently) from $[0,1]$. Matrices with sum of entries smaller than 1 are discarded to avoid highly disconnected graphs.} and for each $n$, we test for $m\in \{10,50,100\}$. Note that each graph here represents an objective, so for a fixed $n$, we generate $m$ Kronecker graphs to get $m$ \textit{max-cover} objectives. For each setting of $n,m$ we evaluate the solution value for the heuristics as $k$ increases and show the average performance over 30 trials for each setting. 
All experiments were performed using MATLAB.	
\begin{figure}[H]
	\centering
	\includegraphics[scale=0.25]{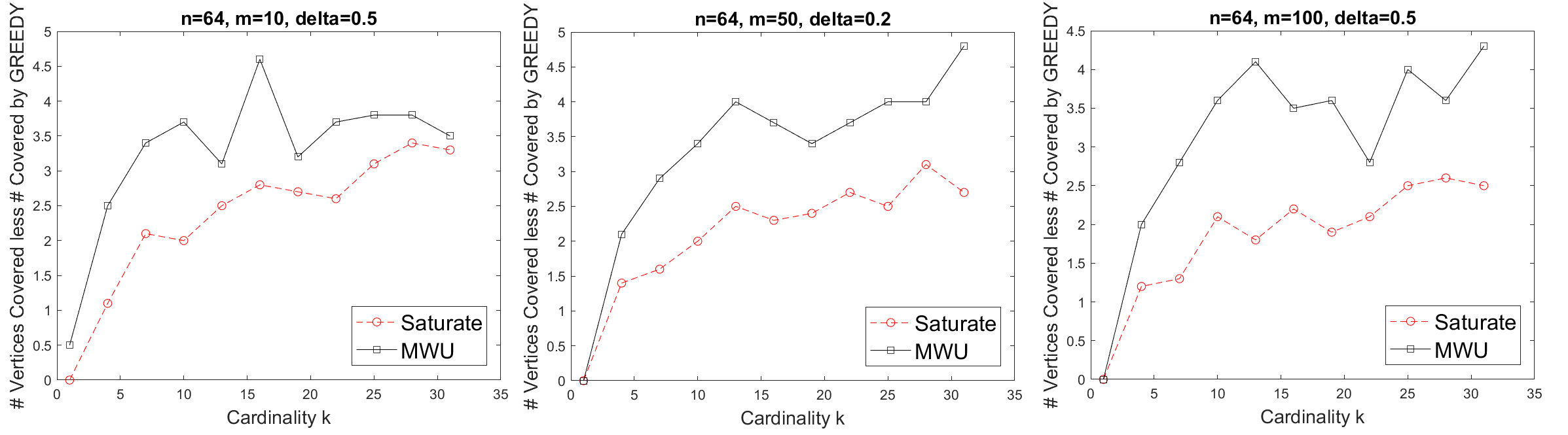}	
	\caption{Plots for graphs of size 64. Number of objectives increases from left to right. The X axis is the cardinality parameter $k$ and Y axis is difference between \# vertices covered by MWU and SATURATE minus the \# vertices covered by GREEDY for the same $k$. MWU outperforms the other algorithms in all cases, with a max.\ gain (on SATURATE) of 9.80\% for $m=10$, 12.14\% for $m=50$ and  16.12\% for $m=100$.}
\end{figure}
\begin{figure}[H]
	\centering
	\includegraphics[scale=0.4]{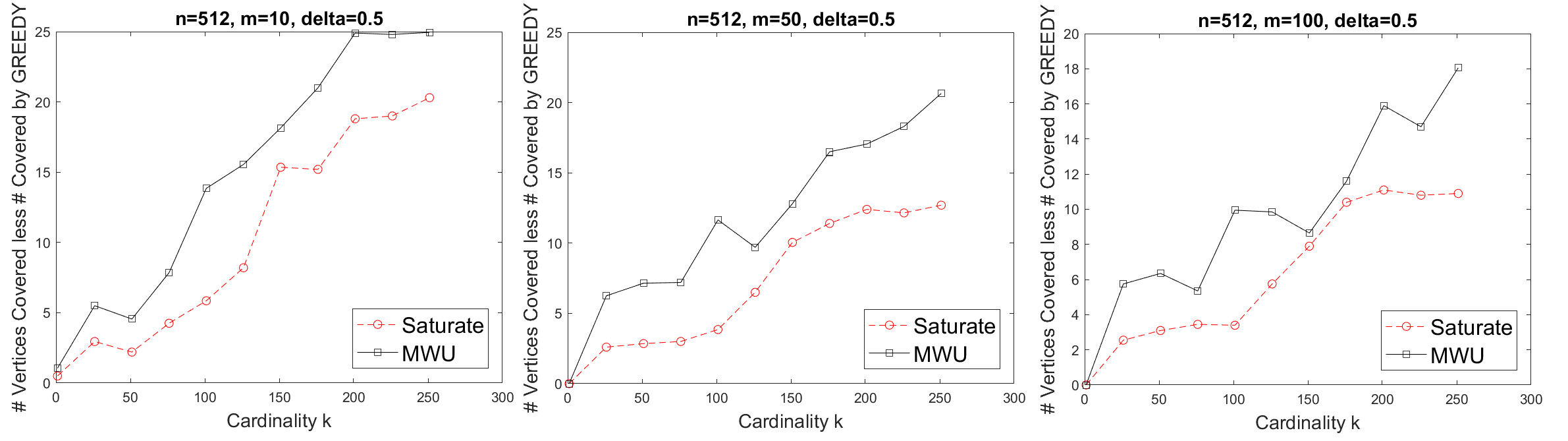}	
	\caption{Plots for graphs of size 512. MWU outperforms SATURATE in all cases with a max.\ gain (on SATURATE) of 7.95\% for $m=10$, 10.08\% for $m=50$ and  10.01\% for $m=100$.}
	
\end{figure}
\begin{figure}[H]
	\centering
	\includegraphics[scale=0.4]{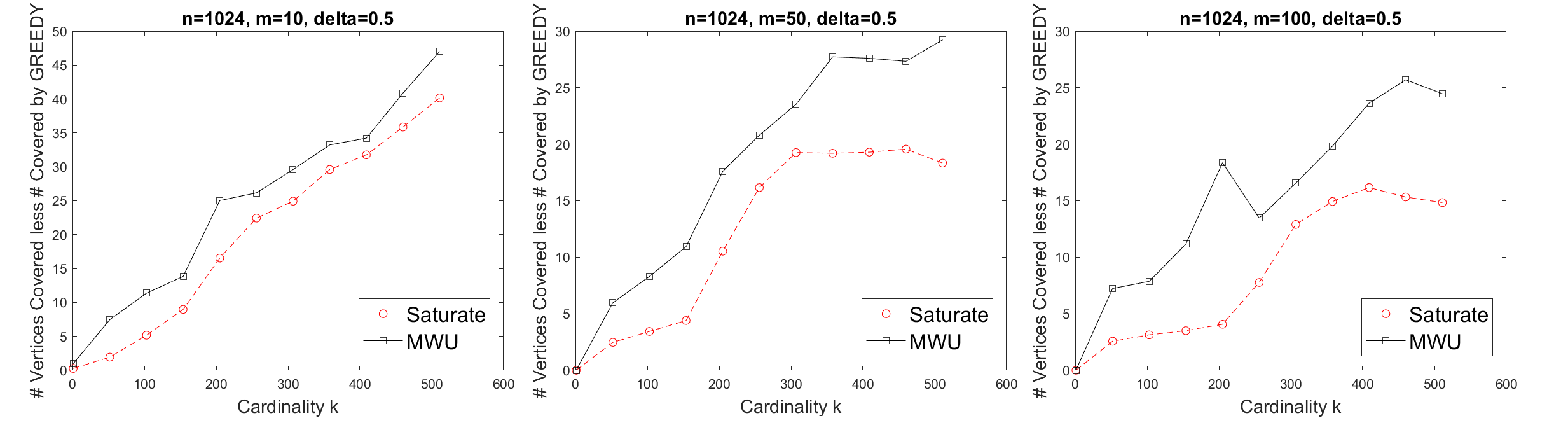}
	\caption{Plots for graphs of size 1024. MWU outperforms SATURATE in all cases, with max. gain (on SATURATE) of 6.89\% for $m=10$, 5.02\% for $m=50$ and  7.4\% for $m=100$.}
\end{figure}
%
\section{Conclusion and Open Problems}
In summary, we consider the problem of multi-objective maximization of monotone submodular functions subject to a cardinality constraint, when $m=o\big(\frac{k}{\log^3 {k}}\big)$. No polynomial time constant factor approximations or strong inapproximability results were known for the problem, though it was known that the problem is inapproximable when $m=\Omega(k)$ and admitted a nearly $1-1/e$ approximation for constant $m$. We showed that when $m=o\big(\frac{k}{\log^3 {k}}\big)$, one can indeed approach the best possible guarantee of $1-1/e$ and further also gave a nearly-linear time $(1-1/e)^2$ approximation for the same. Finally, we established a natural bound on how the optimal solution value increases with increasing cardinality $k$ of the set, leading to a simple deterministic algorithm

A natural technical question here is whether one can achieve approximations right up to $m=o(k)$. Additionally, it also of interest to ask if there are fast algorithms with guarantee closer to $1-1/e$, in contrast to the guarantee of $(1-1/e)^2$ shown here. Perhaps most intriguingly, it is unclear if similar results can also be shown for a matroid constraint.

\begin{APPENDICES}
	\section{On optimizing multilinear extensions $F$}\label{estimatingF}
	\begin{lem}(Theorem 2.2, Lemma 3.2 in \cite{vond1})
		Given a monotone submodular function $f$, let $OPT$ denote the maximum function value for a set of size $d$. To find a solution with value at least $\Big(1-1/e-O(1/d)\Big)OPT$, it suffices to be able to estimate marginals $F({\bf e}_i\mid {\bf y})$ within an additive error of $O\Big(OPT/d^2\Big)$, for any given ${\bf y}\in[0,1]^n$ and $i\in N$. 
		
		For a given pair $i,{\bf y}$, this can be achieved by 
		independently sampling $s=O(d^4 \ln n)$ random sets $R(1),\cdots R(s)$, with elements independently included in each set according to distribution ${\bf y}$. Then using a standard concentration inequality (Theorem 2.2 in \cite{vond1}), the value $\frac{1}{s}\sum_{j=1}^s \big[f(R(j)+i)-f(R(j))\big]$ is within an additive factor $O\big(OPT/d^2\big)$ of $F({\bf e}_i\mid {\bf y})$, for all such marginals computed during the execution of the algorithm w.h.p.
	\end{lem}
{\color{black}\section{Proof of Lemma \ref{needed}}\label{appx:need}
Observe that by definition of new stage 1, on termination we have for every element $e\not\in S_1$, there does not exist any function $f_i\not\in M_1$ such that $f_i(e|S_1)\geq \varepsilon^3 \left(V_i-f_i(S_1)\right).$ It remains to show that $|S_1|\leq m/\varepsilon^3$. 

Let $S_1(j)$ denote intermediate set in stage 1 after $j$ elements have been added. When the $j+1$th element, say $e$, is added we have, $f_i(e|S_1(j))\geq\varepsilon^3 \left(V_i-f_i(S_1(j))\right)$ for some $i$ with $f_i(S_1(j))<(1-1/e)V_i$. 
We show that any given function $f_i$ can help satisfy the criterion for adding a new element at most $1/\varepsilon^3$ times. This immediately implies that $|S_1|\leq m/\varepsilon^3$. To prove the claim, fix an arbitrary $f_i$ and consider the first $1/\varepsilon^3$ steps $\{j_1,\cdots,j_{1/\varepsilon^3}\}$ where an element is added with marginal value at least $\varepsilon^3 \left(V_i-f_i(S_1(j))\right)$ for function $f_i$. After adding the first such element we have, 
\begin{equation}\label{stp1}
	V_i-f_i\left(S_1(j_1)\right) \leq (1-\varepsilon^3) \left(V_i-f_i\left(S_1(j_1-1)\right)\right)\leq (1-\varepsilon^3)V_i. 
	\end{equation}
We claim that after $k$ such elements are added we have set $S_1(j_k)$ with, \[V_i-f_i\left(S_1(j_k)\right)\leq (1-\varepsilon^3)^{k} V_i.\]
The base case for $k=1$ follows from \eqref{stp1}. The rest follows via induction by noticing that after adding the $k+1$th element we have,
\[V_i-f_i\left(S_1(j_{k+1})\right)\leq (1-\varepsilon^3) \left(V_i-f_i\left(S_1({j_{k+1}-1})\right)\right)\leq(1-\varepsilon^3)\left(V_i-f_i\left(S_1({j_{k}})\right)\right). \]
 Thus, after $1/\varepsilon^3$ steps ending with the set $S_1(j_{1/\varepsilon^3})$, we have,  $V_i-f_i(S_1(j_{1/\varepsilon^3}))\leq (1-\varepsilon^3)^{1/\varepsilon^3} V_i<V_i/e$, for any $1-\varepsilon^3\in(0,1)$. Thus, $f_i(S_1(j_{1/\varepsilon^3}))>(1-1/e)V_i$. Therefore, any given $f_i$ can help satisfy the criterion for adding a new element at most $1/\varepsilon^3$ times.

}
\section{Alternative proof of Lemma \ref{key}}\label{keyred}
\begin{proof}{Proof}
	 With abuse of notation we use ${\bf x}_{X^t}$ and $X^t$ interchangeably.
	Let ${\bf x}=\sum_{t=1}^{T}  X^t/T$ and w.l.o.g., assume that sets $X^t$ are indexed such that $f(X^j)\geq f(X^{j+1})$ for every $j\geq 1$. Further, let $f(X^t)/T=a^t$ and $\sum_t a^t=A$. 
	
	Recall that $F({\bf x})$ can be viewed as the expected function value of the set obtained by independently sampling element $j$ with probability $x_j$. Instead, consider the alternative random process where starting with $t=1$, one samples each element in set $X^t$ independently with probability $1/T$. The random process runs in $T$ steps and the probability of an element $j$ being chosen at the end of the process is exactly $p_j=1-(1-1/T)^{Tx_j}$, independent of all other elements. Let ${\bf p}=(p_1,\dots,p_n)$, it follows that the expected value of the set sampled using this process is given by $F({\bf p})$. Observe that for every $j$, $p_j\leq x_j$ and therefore, $F({\bf p})\leq F({\bf x})$. Now in step $t$, suppose the newly sampled subset of $X^t$ adds marginal value $\Delta^t$. From submodularity we have, $\mathbb{E}[\Delta^1]\geq\frac{f(X^1)}{T}= a^1$ and in general, $\mathbb{E}[\Delta^t]\geq\frac{f(X^t)-\mathbb{E}[\sum_{j=1}^{t-1} \Delta_j]}{T}\geq a^t-\frac{1}{T}\sum_{j=1}^{t-1}\mathbb{E}[\Delta^j]$.
	
	To see that $\sum_t \mathbb{E}[\Delta^t]\geq (1-1/e)A$, consider a LP where the objective is to minimize $ \sum_t \gamma^t$ subject to $b^1\geq b^2\dots\geq b^T\geq 0$; $\sum b^t=A$ and $\gamma^t\geq b^t-\frac{1}{T}\sum_{j=1}^{t-1}\gamma^j$ with $\gamma^0=0$. Here $A$ is a parameter and everything else is a variable. Observe that the extreme points are characterized by $j$ such that, $\sum b^t= A$ and $b^t=b^1$ for all $t\leq j$ and $b^{j+1}=0$. For all such points, it is not difficult to see that the objective is at least $(1-1/e)A$. Therefore, we have $F({\bf p})\geq (1-1/e)A=(1-1/e)\sum_t f(X^t)/T$, as desired. 
	
	\end{proof}
\end{APPENDICES}
\ACKNOWLEDGMENT{
 Thanks to Xiaoyun Fu, Pavan Aduri, and Samik Basu for pointing out an error in Theorem \ref{extension} in previous versions of this paper. The author thanks James B. Orlin and anonymous referees for their insightful comments and feedback, and Mohit Singh for a fruitful discussion that led to further simplification of the results. Part of this work was supported by ONR grant N00014-17-1-2194.
}


	\bibliographystyle{informs2014}
	\bibliography{new}

\end{document}